\documentclass[12pt, draftclsnofoot, onecolumn]{IEEEtran}
\usepackage{blindtext}
\usepackage{graphicx}
\usepackage[table]{xcolor}

\usepackage{cite}

\usepackage{qtree}
\usepackage{multirow}
\usepackage{amsmath, amsthm}
\usepackage{algorithm}
\usepackage{algorithmic}

    \newcommand{\CASE}[1]{\STATE \textbf{case} #1\textbf{:} \begin{ALC@g}}
    \newcommand{\ENDCASE}{\end{ALC@g}}
    
    \newcommand{\DEFAULT}{\STATE \textbf{default:} \begin{ALC@g}}
    \newcommand{\ENDDEFAULT}{\end{ALC@g}}
    \newcommand{\DEFAULTLINE}[1]{\STATE \textbf{default:} }
\usepackage{amssymb}
\usepackage{mathrsfs}

\newcolumntype{C}[1]{>{\centering\let\newline\\\arraybackslash\hspace{0pt}}m{#1}}

\usepackage{setspace,tikz}

\usepackage{pgfplots}
\pgfplotsset{compat=newest}

\newtheorem{proposition}{Proposition}

\begin{document}

\title{Multi-Version Coding}

\author{Majid~Khabbazian
  \thanks{%
    M. Khabbazian is with the Department of Electrical and Computer Engineering, University of Alberta, Edmonton, Canada
    (Email: {mkhabbazian@ualberta.ca}).}}

\maketitle

\begin{abstract}
  We derive a simple lower bound for the multi-version coding problem formulated in~\cite{WangC14a}.
  We also propose simple algorithms that almost match the lower bound derived.
  Another lower bound is proven for an extended version of the multi-version coding problem introduced in~\cite{WangC14b}.
\end{abstract}

%\begin{IEEEkeywords}
%IDNC, Broadcast
%\end{IEEEkeywords}

\section{Introduction}
  We study the multi-version coding problem formulated by Wang and Cadambe~\cite{WangC14a}.
  In this problem, there is a distributed storage system with $n$ servers, 
  and a client with $v$ independent message versions.
  The informal description of the problem is as follows.
  Every time, the client uploads one version (starting with version 1) by connecting to
  these $n$ servers.
  Because of network failures, a version may not reach all the servers.
  However, when a version is reached/received by a server, the server stores some information 
  about that message version (not necessarily the whole message), and perhaps modifies the information already stored.
  For example, in the \emph{replication strategy}, when a version reaches a sever, the server stores the whole version 
  and deletes any version stored before.
  
  Let $c$, $1\leq c\leq n$ be an integer.
  The multi-version coding problem requires that the client should be able to download a version $i$, $1\leq i\leq v$, 
  by connecting to any set of $c$ servers $\mathcal{S}$, if version $i$ is the latest version reached by all the servers in $\mathcal{S}$.
  The objective of the problem is to minimize the worst-case storage cost per server, defined as the size of server's storage
  divided by the size of message (assuming that all versions have the same size).
  
  By the above definition, the storage cost of the simple replication strategy is one.
  When $c<v$, a better strategy, as stated in~\cite{WangC14a}, is to use an $(n,c)$ MDS code for each version.
  Using this approach, the worst-case storage cost is $\frac{v}{c}$.
  Interestingly, it was shown that the cost of $\frac{v}{c}$ can be slightly reduced for $v=2$, and $v=3$, to
  $\frac{2c-1}{c^2}$, and $\frac{3c-2}{c^2}$, respectively~\cite{WangC14a}.
  The authors of~\cite{WangC14a} also proved a lower bound of $1-(1-\frac{1}{c})^v$ for the worst-case storage cost, hence concluded that
  when the number of versions $v$  approaches infinity, the \emph{replication strategy} is close to optimal.
  Their lower bound also indicates that for small values of $v$, MDS codes are almost optimal.
  
  In this work, we prove a new lower bound on the worst-case storage cost.
  Our lower bound shows that when $v> c$, the \emph{replication strategy} is optimal. 
  %Also, it shows that the MDS codes are almost optimal when $v\leq c$.
  We propose two algorithms based on erasure codes that can achieve near optimal storage cost for any $v\leq c$.
  This answers an open question raised in~\cite{WangC14a} on designing codes for moderate values of $v$.

\section{Lower Bound}

\begin{proposition}
  The worst-case storage cost of the multi-coding problem is lower bounded by $\min(1, \frac{v}{c+1})$.\\
  Note that $\frac{v}{c+1}\approx \frac{vc-(v-1)}{c^2}$.
\end{proposition}
\begin{proof}
  Suppose $v\leq c$, and $n=c+1$.
  Assume that server $i$, $v+1\leq i\leq c+1$ were reached by  all the $v$ versions.
  Also, assume that server $i$, $1\leq i\leq v$, were reached by all the $v$ versions except version $i$.
  Let $S_i$, $1\leq i \leq v$, be the subset of servers including all servers except $i$.
  Note that, for every $1\leq i\leq v$, $|S_i|=c$, and the latest version reached by all server in $S_i$ is $i$.
  Therefore, we must be able to retrieve version $i$, $1\leq i \leq v$, by connecting to $S_i$.
  This implies that the set of all $c+1$ servers must contain information about all $v$ versions. 
  Hence, the storage cost per server must be at least $\frac{v}{c+1}$, in this setting.
  Note that, by partitioning the set of servers to parts of size $c+1$, this argument is easily generalized to the case where $c+1|n$ .
\end{proof}

\section{Simple Near-Optimal Multi-Version Coding Algorithms}
   Following we informally describe two multi-version coding algorithms.
   The proposed algorithms assure that at each step of the process the storage cost per server does not exceed the maximum storage cost.
   Also the information stored for one version does not need to increase when other versions arrive. 
    
\subsection{First Algorithm}

  The first algorithm uses a $(n, c+1)$ MDS code  for versions $1\leq i \leq v-1$, and
  a $(n,c)$ MDS code for version $v$ (the last version).
  Suppose the size of each version is $B$ bits.
  Upon receiving a version $i$, $1\leq i\leq v-1$, a server stores $\frac{2B}{c+1}$ bits of coded information for that version,
  and reduces the information stored for version $i-1$ from $\frac{2B}{c+1}$  to $\frac{B}{c+1}$ (if version $i-1$ has received before).
  Every server that receives version $v$, that is the latest version, just stores $\frac{B}{c}$ bits of coded information for it.
  Now, first note that, in the worst case, the total storage cost of a server is $(v-1)\frac{B}{c+1}+\frac{B}{c}$, 
  which is less than $\frac{vc-(v-1)+1}{c^2}B$.
  Second, if version $i$, $1\leq i\leq v-1$ is the latest version reached by a set of $c$ servers, then the total information about version $i$ 
  stored in those servers is at least $(c-1)\frac{B}{c+1}+\frac{2B}{c+1}=B$, 
  where $\frac{2B}{c+1}$ is due to the fact that at least one of those servers has not been reached by version $i+1$.
  If version $v$ is the latest version reached by the servers, then the total information of version $v$ at the servers is clearly $c\cdot\frac{B}{c}=B$.
\subsection{Second Algorithm}
  The second algorithm slightly improves the storage cost of the first algorithm to $\frac{vc-(v-1)}{c^2}B$, which almost matches the lower bound proven.
  Here, we just explain how storage is assigned for each version on a server.
  Using coding we can easily guarantee that a version is retrievable from a set of servers 
  as long as the sum of storages assigned to that version by the set of servers is at least $B$ bits.
  
  In the second algorithm, upon receiving the first version, a server stores $\frac{vc-(v-1)}{c^2}B$ bits of information.
  When another version is received, the server deletes $\frac{B}{c}$ bits of information of the first version, and stores $\frac{B}{c}$ bits 
  of information of the version received. 
  Now consider a set $\mathcal{S}$ of $c$ servers. 
  If the latest version reached by all servers in $\mathcal{S}$ is $i>1$, then each server has $\frac{B}{c}$ bits of information of that version,
  so the latest version can be decoded.
  If the latest version is the first version, then the total information of the first version stored in all servers in $\mathcal{S}$ is at least
  \[
    c\cdot \frac{c-(v-1)}{c^2}B+ (v-1)\frac{B}{c}=B,
  \]
  where the term $ (v-1)\frac{B}{c}$ is due to the fact that versions $2,3,\ldots v$ are not the latest versions reached,
  hence the servers that miss those versions have deleted $\frac{B}{c}$ less bits of information from their first version for each missing version.
  
\section{Extended Multi-Coding problem}
  In the original multi-coding problem, the latest version reached by a set of $c$ servers should be decodable.
  This can be relaxed, as explained in~\cite{WangC14b}, by requiring the latest version or any later version to be decodable.
  In~\cite{WangC14b}, it was shown that the storage cost of the extended multi-coding problem is strictly less than that in the original problem.
  The following lower bound on the worst storage cost per server was proven in~\cite{WangC14b}:
  \[
    \text{storage cost }\geq
    \begin{cases}
      \frac{2}{c+1} & \text{if $c$ is odd,}\\
      \frac{2(c+1}{c(c+2)} & \text{if $c$ is even.}
    \end{cases}
  \]
  
  Note that the above lower bound does not depend on $v$.
  Here, we prove a lower bound that is an increasing function of $v$.
  In particular, we show that the storage cost of the extended multi-cast problem is lower bounded by $\frac{v}{c+v-1}$.
  Then, we show that the bound is tight when $c=vq+1$ for some non-negative integer $q$.
  
\begin{proposition}
  The worst-case storage cost for the extended multi-coding problem is at least $\frac{v}{c+v-1}$.
\end{proposition}
\begin{proof}
  %Here is the proof sketch using non-technical words.\\
  The set of versions reached by a server is called the \emph{profile} of the server.
  To prove the proposition, we construct $m$ profiles, iteratively.
  Then, we consider a set of $m$ servers each with one of those profiles, and argue on the minimum amount of information
  those servers should have, collectively.
  In the following, we represent a profile with a binary vector of size $v$, where a ``1'' in coordinate $i$, $1\leq i \leq v$ implies reception 
  of version $i$.  
  Note that a server with a ``1'' in coordination $i$ in its profile has not necessarily stored any information about version $i$.
%  Note that a ``1'' in coordination $i$ only indicate that the server has received version $i$.
%  It does not imply that the server has stored information about version $i$. 
%  The amount of information of version $i$ stored by a server can be as low as zero, even when the server receives version $i$.
   A ``0'' in coordinate $i$, however, indicates that version $i$ has not been received, therefore the server will have no information about version $i$.
   
  The construction of profiles is performed iteratively starting with  profile $p_1=(1,1,1, \ldots, 1)$, 
  that is the profile of a server that has received all the versions.
  Let $p_i$ be the profile constructed in the $i$th iteration.
  To construct $p_{i+1}$, we initially set $p_{i+1}$ to $p_i$.
  If the set of $i+1$ servers with profiles $p_1, \ldots, p_{i}, p_{i+1}$ have at least $B$ bits of information about
  a version $j$, then we set he coordinate $j$ in vector $p_{i+1}$ to zero.
  We repeat this process of nullifying coordinates until the set of $i+1$ servers with profiles $p_1\ldots, p_{i+1}$ do not have enough information 
  (that is $B$ bits of information) about any version.
  We terminate if $p_{i+1}$ is a zero vector, and set $m$ to $i$.
  
  First, we show that $m\leq c-1$.
  By contradiction, assume $m\geq c$.
  Then, there must be a coordinate $j$ which is equal to one in all the profiles $p_1, p_2, \ldots, p_m$.
  This is a contradiction, since, in that case, the set of $c$ servers with profiles $p_1, p_2, \ldots, p_c$
  have at least one common version, hence they can collectively decode at least one version 
  (that is, they must have enough information about at least one version).
  
  Next we show that, for any version $u$, the set of $m$ servers with profiles $p_1, p_2, \ldots, p_m$
  collectively have at least $B-t$ bits of information, where $t$ is the maximum storage cost per server.
  Fix any version $u$.
  Let $1\leq j\leq m$ be the first iteration in the profile construction process where the coordinate corresponding to version $u$ is set to zero.
  This implies that there is a profile $p$ such that the set of $j$ servers with profiles $p_1, \ldots, p_{j-1}, p$ have enough information about version $u$.
  Note that the maximum amount of information per server for version $u$ is $t$.
  Therefore, the set of $j-1$ servers with profiles $p_1, \ldots, p_{j-1}$ must collectively have at least $B-t$ bits of information about version $u$.
  Since this holds for any version, the servers with profiles $p_1, \ldots, p_m$ must have at least $v(B-t)$ bits of information
  about all $v$ versions.
  The maximum storage cost per server is $t$.
  Therefore, we must have $\frac{v(B-t)}{m}\leq t$, thus $\frac{v(B-t)}{c-1}\leq t$, hence $t\geq \frac{v}{c+v-1}B$.

%  To prove the lower bound, we construct $p$, $p\leq c+v-1$ profiles and prove than any set of $p=c+v-1$ servers with those profiles 
%  must collectively have information about all the $v$ versions.
%  In the following, we represent a profile with a binary vector of size $v$, in which a one in location $i$, $1\leq i \leq v$ implies reception 
%  of version $i$.
  
\end{proof}

 Suppose each server only stores information about the latest version received.
 Without loss of generality, suppose $B=1$.
 Assume that the amount of storage assigned to the latest version is $\frac{1}{\lceil\frac{c}{v} \rceil}$.
 Consider a set of $c=vq+d$ servers, where $q$ is a non-negative integer and $1\leq d\leq v-1$.
 Assume that each server has received at least one version. 
 This this is a more general assumption compared to the problem's assumption, 
 which only considers the set of $c$ servers that have at least one common version.
 Since each server has received at least one version, there must be at least $q+1$ servers with identical latest versions.
 Each of those servers has assigned $\frac{1}{\lceil\frac{c}{v} \rceil}$ storage to their latest version.
 Therefore, the total amount of storage assigned to that version is
 \[
   (q+1)\frac{1}{\lceil\frac{c}{v} \rceil}=1
 \]
 When $d=1$, that is when $c=vq+1$, we get
 \[
   \frac{1}{\lceil\frac{c}{v} \rceil}=\frac{1}{\frac{c+(v-1)}{v}}=\frac{v}{c+v-1}.
 \]
 For instance, when $c=v+1$, it is possible to get the optimal storage cost of $\frac{1}{2}$, which is almost $50\%$ lower than
 the minimum storage cost achievable in the original multi-version coding problem.
 
 We remark that, under the general assumption mentioned above, the storage cost of $\frac{1}{\lceil\frac{c}{v} \rceil}$ is
 optimal.
 The reason is as follows:
 Consider $v$ groups of servers, each group with $\lceil\frac{c}{v} \rceil$ servers in it.
 Note that the total number of servers in all groups is at least $c$.
 Suppose every server in group $i$, $1\leq i\leq v$ has received only version $i$.
 If the storage cost per server is less than $\frac{1}{\lceil\frac{c}{v} \rceil}$, for any version $i$, the total information about version $i$ stored 
 by servers in all the $v$ groups will be less than one.
 In this case, no version can be decoded by the above set of $v\cdot \lceil\frac{c}{v} \rceil\geq c$ servers.

\section{conclusion}
  Based on the first lower bound derived, the simple replication strategy is optimal if the number of versions is more than $c$.
  For smaller number of version, there is a simple strategy based on MDS codes that can almost achieve the lower bound derived.
  Our second lower bound improves the lower bound on the storage cost of the extended multi-version coding problem proposed in~\cite{WangC14b}.
  It is also tight for many values of $v$.

\end{document}